\documentclass[a4paper,english]{lipics-v2018}
\usepackage{cite}
\usepackage{microtype}
\usepackage{todonotes}

\nolinenumbers

\EventEditors{Hiro Ito, Stefano Leonardi, Linda Pagli, and Giuseppe Prencipe}
\EventNoEds{4}
\EventLongTitle{9th International Conference on Fun with Algorithms (FUN 2018)}
\EventShortTitle{FUN 2018}
\EventAcronym{FUN}
\EventYear{2018}
\EventDate{June 13--15, 2018}
\EventLocation{La Maddalena, Italy}
\EventLogo{}
\SeriesVolume{100}
\ArticleNo{21} 

\newcommand{\oeis}[1]{\href{https://oeis.org/#1}{#1}}

\usepackage{listings}
\usepackage{color}
\theoremstyle{definition}
\newtheorem{observation}[theorem]{Observation}

\newcommand{\agg}{\operatorname{AGG}}
\newcommand{\sgg}{\operatorname{Single}}
\newcommand{\tgg}{\operatorname{Total}}
\renewcommand{\ggg}{\operatorname{GreedyCoins}}

\begin{document}
\title{Making Change in 2048}

\author{David Eppstein}{Computer Science Department, University of California, Irvine}{eppstein@uci.edu}{}{Supported in part by NSF grants  CCF-1618301 and CCF-1616248.}
\authorrunning{David Eppstein}
\Copyright{David Eppstein}

\subjclass{\ccsdesc[500]{Theory of computation~Discrete optimization}}
\keywords{2048, change-making problem, greedy algorithm, integer sequences, halting problem}

\maketitle

\begin{abstract}
The 2048 game involves tiles labeled with powers of two that can be merged to form bigger powers of two; variants of the same puzzle involve similar merges of other tile values. We analyze the maximum score achievable in these games by proving a min-max theorem equating this maximum score (in an abstract generalized variation of 2048 that allows all the moves of the original game) with the minimum value that causes a greedy change-making algorithm to use a given number of coins. A widely-followed strategy in 2048 maintains tiles that represent the move number in binary notation, and a similar strategy in the Fibonacci number variant of the game (987) maintains the Zeckendorf representation of the move number as a sum of the fewest possible  Fibonacci numbers; our analysis shows that the ability to follow these strategies is intimately connected with the fact that greedy change-making is optimal for binary and Fibonacci coinage. For variants of 2048 using tile values for which greedy change-making is suboptimal, it is the greedy strategy, not the optimal representation as sums of tile values, that controls the length of the game. In particular, the game will always terminate whenever the sequence of allowable tile values has arbitrarily large gaps between consecutive values.
\end{abstract}

\section{Introduction}

The solitaire game 2048 was developed in 2014 by Gabriele Cirulli, based on another game called Threes developed earlier in 2014 by Asher Vollmer~\cite{Her-ICGA-14}.
It is played on a 16-cell square grid, each cell of which can either be empty or contain a tile labeled with a power of two. In each turn, a tile of value 2 or 4 is placed by the game software on a randomly chosen empty cell. The player then must tilt the board in one of the four cardinal directions, causing its tiles to slide until reaching the edge of the board or another tile. When two tiles of equal value slide into each other, they merge into a new tile of twice the value. The game stops when the whole board fills with tiles, and the goal is to achieve the highest single tile value possible.  \autoref{fig:8192} shows the state of the game after approximately $4000$ moves, when a tile with value $8192$ has been reached.

As most players of the game quickly learn, it is not possible to keep playing a single game of 2048 forever. At any step of the game, there must be at least one tile for each nonzero bit in the binary representation of the total tile value. For total tile values just below a large power of two, the number of ones in the binary representation is similarly large, eventually exceeding the number of cells in the board.

\begin{figure}[t]
\centering\includegraphics[scale=0.5]{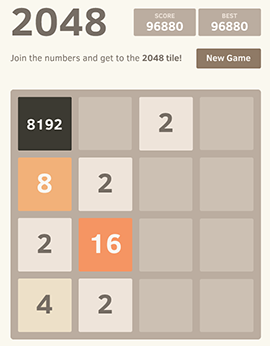}
\caption{A state in the game 2048 in which a tile of value 8192 has been reached}
\label{fig:8192}
\end{figure}

But other variants of 2048 use different tile values than powers of two. Threes uses the sequence of numbers that are either powers of two or three times a power of two:
\[
1, 2, 3, 4, 6, 8, 12, 16, 24, 32, 48, 64, \dots
\]
(It also restricts tile merges to pairs of tiles whose values are equal or differ by a factor of two.) Fives uses 2, 3, and powers of two times 5, giving the sequence of allowable values~\cite{LanUno-FUN-16}
\[
1, 2, 3, 5, 10, 20, 40, 80, 160, 320, 640, \dots
\]
Another variant, called 987, uses as its tile values the Fibonacci numbers,
\[
1, 2, 3, 5, 8, 13, 21, 34, 55, 89, 144, 233, \dots
\]
We can find analogous ad-hoc arguments for why these games must terminate, but can we generalize them to arbitrary systems of tile values? If we define a 2048-like game with a set $S$ as its tile values, is the length of the game and the maximum value that can be achieved controlled, as it is for binary numbers, by the lengths of the shortest representations of arbitrary numbers as sums of members of $S$?

For instance, suppose that we allow any \emph{practical number} as a tile value,
and any merge of two tiles that produces another practical number.
The practical numbers are defined by the property that, for a practical number~$n$, every integer $m<n$ can be expressed as a sum of distinct divisors of $n$. Their sequence begins
\[
1, 2, 4, 6, 8, 12, 16, 18, 20, 24, 28, 30, 32, \dots
\]
There are many more practical numbers than powers of two, and the practical numbers behave in many ways like the prime numbers. In particular, analogously to Goldbach's conjecture for the prime numbers, every even integer can be expressed as a sum of two practical numbers~\cite{Mel-JNT-96}, and therefore every integer can be expressed as a sum of three practical numbers. Because we can express every tile value using a bounded number of practical-number tiles,
does the practical-number variant of 2048 go on forever?

Alternatively, suppose we use $3$-smooth tile values, the numbers whose only prime factors are two or three:
\[
1, 2, 3, 4, 6, 8, 9, 12, 16, 18, 24, 27, 32, 36, \dots
\]
Because the number of distinct $3$-smooth numbers in the range from $1$ to $n$ is only $O(\log^2 n)$, an information-theoretic argument shows that some numbers in this range will require $\Omega(\log n/\log\log n)$ terms in their shortest representation as a sum of $3$-smooth numbers. Therefore, for a game using these tile values to last for $n$ moves, it must use a game board that has at least  $\Omega(\log n/\log\log n)$ cells. Is this analysis tight?

\subsection{New results}
In this paper we show that the answer to these questions is no. 2048-like games are not controlled by the shortest representations of numbers as sums of tile values, but rather by their \emph{greedy representations}, representations generated by a greedy heuristic for the problem of making change using the smallest number of coins from a given coinage system. For the powers of two, the Fibonacci numbers, and the numbers used by Threes and Fives, these greedy representations coincide with the shortest representations, but that is not true for many other natural sets of numbers including the practical numbers and the $3$-smooth numbers. The lengths of greedy representations, in turn, are controlled by the lengths of the gaps between consecutive tile values.

As a consequence, we show that whenever a sequence of numbers has arbitrarily large gaps, the 2048-like game based on those numbers must terminate with a finite limit on its number of moves and on its largest achievable tile value.
For instance, because the practical numbers have inverse-logarithmic density (analogously to the prime number theorem for the density of the prime numbers)~\cite{Wei-QJM-15}, they have arbitrarily large gaps and the game based on them terminates, albeit much more slowly than for the powers of two.

\subsection{Related work}

2048 has been the subject of much past research. Its past investigations include studies of its computational complexity~\cite{Meh-14,AbdAchDas-15,AbdAchDas-FUN-16,LanUno-FUN-16}, artificial intelligence based game strategies~\cite{RodLev-CIG-14,SzuJas-CIG-14,OkaMat-CG-16,YehWuHsu-TCIAIG-17,Mat-ACG-17,Jas-TCIAIG-17}, computer science education~\cite{Nel-JCSC-15}, and computer-human interaction~\cite{PorLeeEge-CHI-15}.

\section{Simplification through abstraction}

Several features of 2048 and its variants complicate their analysis,
possibly making its game play more interesting but
without (it appears) greatly affecting the questions we wish to study,
on how long a game can last or which tile values can be achieved.

\begin{description}
\item[Board geometry.]
The cells of the 2048 board are arranged in a square grid,
which controls both the sliding movement of the tiles across the board and the pairs of tiles that can become adjacent to each other and merge. Much of the strategy of the game involves linearizing this two-dimensional arrangement of cells by finding a zigzag path that covers all the cells of the grid and playing in such a way that tiles move and merge with each other only along this path.

\item[Restricted tile merges.]
In some variations of 2048, such as Threes, certain pairs of tiles cannot merge even when their summed value would be an allowable tile value. For instance, in Threes, the merge $3+1=4$ is not allowed; only pairs of tiles with the same value or with one twice the other can merge. Even in 2048, only pairs of tiles, and not larger combinations of tiles, are allowed to merge.

\item[Unknown or random future events.]
In 2048, the next tile could either have value 2 or 4. In most cases this causes little change to game play, because a tile of value 4 is not significantly different than two consecutive tiles of value 2 that then became merged, but it can interact with the board geometry to cause tiles to become out of position, making continued play more difficult. And in many of these games, the location of each newly placed tile could be any previously-open cell. These unknowns make the game nondeterministic, and complicate the definition of the longest play or highest achievable tile value: do we mean the worst case (the best that a player could achieve against a malicious adversary), best case (the best one could hope to achieve against repeated play with a random adversary), or some kind of probabilistic analysis that determines the distribution or expected value of scores?
\end{description}

To avoid these complications, we define a class of variants of 2048 in which they are eliminated.

\begin{definition}[abstract generalized games]
Given a set $A$ of allowable tile values, an initial element $a\in A$ (usually $a=1$), and a number $n$ of cells, we define the \emph{abstract generalized 2048 game} for $A$ and $a$ to be
a solitaire game in which there are $n$ indistinguishable cells, each of which can either be empty or contain a tile with a value in $A$. We define a \emph{position} of the game to be an assignment of either a tile with a value in $A$ or no tile to each cell of the game. The \emph{initial position} of the game is a position in which all cells are empty.
Starting from the initial position, each step of the game consists of the following actions:
\begin{itemize}
\item The player chooses any empty cell, and a tile of value $a$ is placed into that cell.
\item The player may choose to merge any sets of non-empty cells whose total value belongs to $A$ into a single tile, which is placed on a single cell from its set. The remaining cells in each chosen set become empty.
\end{itemize}
The game ends when, after one of these steps, all cells are nonempty. When this happens, there would be nowhere to place the new tile of value~$a$ in the next step.

We denote the abstract generalized 2048 game on $n$ cells with tile value set $A$ and initial tile value $a$ by $\agg(n,A,a)$ or (when $a=1$) by $\agg(n,A)$.
\end{definition}

\begin{observation}[simulation by abstract games]
With the possible exception of the value of each newly placed tile,
each action in 2048, Threes, Fives, or 987 can be simulated by a corresponding action
in the abstract generalized 2048 game with the same set of tile values and the same number of cells.
Therefore, any upper bound on the number of moves or maximum tile value in the abstract generalized 2048 game provides a valid upper bound for the number of moves or maximum tile value in the corresponding sliding-tile game.
\end{observation}

\section{Optimal strategy in the abstract game}

The abstract generalized 2048 game eliminates the complications of board geometry, tile position, and sliding mechanics from the game, making its analysis much simpler. As a consequence, we can characterize the optimal strategies in this game. We begin by describing some helpful move-ordering principles.

\begin{definition}[eager sequences]
We say that a sequence of steps in $\agg(n,A)$ is \emph{eager}
if each merge of tiles is performed in the first step at which all of the tiles to be merged have their merged values, rather than delaying the merge until some later step.
\end{definition}

\begin{observation}[all sequences can be made eager]
\label{obs:eager}
If a position in $\agg(n,A)$ can be reached by a sequence of steps, it can be reached by an eager sequence of steps.
\end{observation}

\begin{lemma}[single-tile-first strategy]
\label{lem:stf}
Let $P$ be a position in $\agg(n,A)$ that can be reached by a sequence of steps from the initial position. Then there exists a non-empty cell $c$ of value $v$ in $P$, and a sequence of steps that reaches $P$ from the initial position, with the following structure:
\begin{itemize}
\item First, perform a sequence of steps that reaches the position $P'$, where $P'$ has a tile of value $v$ in cell $c$ and $n-1$ empty cells.
\item Next, perform a sequence of steps in the game $\agg(n-1,A)$, using only the cells that are empty in position $P'$, to reach the position $P''$ in that game corresponding to position $P$ in $\agg(n,A)$ (the position formed from $P$ by removing one cell of value $m$).
\end{itemize}
\end{lemma}

\begin{proof}
Let $S$ be an eager sequence of steps that produces position $P$. By assumption, $S$ exists, and we may assume by Observation~\ref{obs:eager} that $S$ is eager.
By running the sequence of steps in $S$ backwards from $P$,
we may determine, for each position reached during the course of sequence $S$, which of its nonempty tiles eventually contribute to each tile of $P$.
By the eager property of $S$, each merge produced in each step involves only tiles that contribute to the same cell as the newly-placed tile in that step.

Let $c$ be the cell of $P$ to which the first newly-placed tile contributes, and let $v$ be the value of the tile in cell $c$ of position $P$. Because the cells are indistinguishable, we may rearrange the cells of $\agg(n,A)$ so that the first newly-placed tile is placed into cell $c$, and so that each subsequent merge step involving this tile places the merged tile back into cell $c$.
After this rearrangement, cell $c$ is always occupied by a tile that contributes to the eventual value in cell~$c$.
We may then separate $S$ into two subsequences of steps, the subsequence $S_1$ of steps whose newly placed tile contributes to $c$ and the subsequence $S_2$ of steps in which the newly placed tile contributes to some other cell of $P$.

Then subsequence $S_1$ may be performed first, before any steps of $S_2$. This change of order causes positions of the game to be empty in $S_1$ that are non-empty in $S$, but those positions do not contribute to $c$ and therefore do not affect what happens in these steps. Performing $S_1$ reaches state $P'$, as described by the statement of the lemma.

The cells other than $c$ form an instance of $\agg(n-1,A)$, and each step of $S_2$ operates only on these cells because at each of these steps, $c$ is occupied by a tile that is unchanged by that step.
Therefore, $S_2$ may be performed on $\agg(n-1,A)$ to reach position $P''$.
Because $c$ is the only nonempty cell after $S_1$ and is unused by $S_2$, it is valid to perform the concatenation of subsequences $S_1S_2$, which reaches state $P$ with the desired step ordering.
\end{proof}

\begin{lemma}[step-by-step reachability for single tiles]
\label{lem:step-by-step}
Let $x>1$ be a tile value in set $A$, and let $y$ be the largest value in $A$ that is less than $x$.
Let $n$ be a positive integer, let $P_x$ be the position consisting of one cell containing a tile of value $x$ and $n-1$ empty tiles, and let $P_y$ be defined in the same way for value $y$.
Then there is a sequence of steps in $\agg(n,A)$ that reaches $P_x$ if and only if the following conditions are both true:
\begin{enumerate}
\item There is a sequence of steps in $\agg(n,A)$ that reaches $P_y$.
\item There is a sequence of steps in $\agg(n-1,A)$ that reaches a position of total value $x-y$.
\end{enumerate}
\end{lemma}

\begin{proof}
We prove separately that both conditions imply reachability of $P_x$, and that reachability of $P_x$ implies both conditions.

\begin{description}
\item[(1 \& 2) $\Rightarrow P_x$:]~\\
Clearly if both conditions are true, then we can use the sequence from the first condition to reach $P_y$, then concatenate the sequence of steps from the second condition to reach a position that includes both $y$ and some other tiles of total value $x-y$, and finally perform a single merge operation to combine all of these tiles to a single tile of value $x$.

\item[$P_x\Rightarrow$ (1):]~\\
Let $S$ be any sequence of steps that reach $P_x$. Then the first $y$ steps of $S$ reach a position of total value $y$, from which $P_y$ can be formed by one more merge operation.

\item[$P_x\Rightarrow$ (2):]~\\
$P_x$ is reachable if and only if we can reach a position $P'$ of total value $x-1$, with at least one empty cell, so that the newly placed tile of the next step creates total value~$x$.
By the single-tile-first strategy (Lemma~\ref{lem:stf}), $P'$ is reachable if and only if there exists $z\in A$, with $0<z<x$,
such that $P_z$ is reachable in $\agg(n,A)$ and the remaining cells of $P'$, of total value $x-z-1$ and with at least one empty cell, are reachable in $\agg(n-1,A)$.
If $y=z$ then one more step in $\agg(n-1,A)$ places a new tile of value 1 in the empty cell and creates a position of total value $x-y$, meeting condition~2. If, on the other hand, $y>z$, then $x-y\le x-z-1$ and
the first $x-y$ steps in $\agg(n-1,A)$ already create a position of total value $x-y$, again meeting condition~2.
\end{description}
\end{proof}

\begin{corollary}[threshold of single-tile reachability]
For every $n$ and $A$ then there exists a value $\sgg(n,A)\in A\cup\{\infty\}$ such that
the positions $P_x$ (with one tile of value $x$ and $n-1$ empty cells) are reachable in $\agg(n,A)$ if and only if $x\in A$ and $x\le \sgg(n,A)$. If $\sgg(n,A)$ is finite, it is the maximum single tile value achievable in game $\agg(n,A)$; if not, all tile values are achievable.
\end{corollary}

\begin{observation}[monoticity of single-tile thresholds]
\label{obs:monsgg}
For all $n>1$ and $A$, $\sgg(n,A)\ge\sgg(n-1,A)$.
\end{observation}

\begin{proof}
If we can reach any single tile value $v$ in game $\agg(n-1,A)$, we can also reach it in $\agg(n,A)$ by ignoring the extra cell.
\end{proof}

\begin{theorem}[characterization of reachable positions]
\label{thm:xreach}
Let $n$ and $A$ be given. Then a position $P$ of $\agg(n,A)$ is reachable by a sequence of steps from its initial position if and only if the sequence of its tile values $v_1,\dots v_n$ (sorted from smallest to largest, with $v_i=0$ if there are at least $i$ empty cells in $P$) satisfies the inequalities $v_i\le\sgg(i,A)$ for all~$i$.
\end{theorem}

\begin{proof}
By applying the single-tile-first strategy (Lemma~\ref{lem:stf}) recursively,
we may decompose $P$ into a sequence of tile values $u_n$ (the single tile used by the strategy to reach $P$),
$u_{n-1}$ (the tile used by applying Lemma~\ref{lem:stf} to the game $\agg(n-1,A)$ after constructing tile $u_n$, $\dots$ padding the sequence with zeros if necessary.
Then by construction $u_n$ is achievable in game $\agg(n,A)$, $u_{n-1}$ is achievable in game $\agg(n-1,A)$, etc., so these values satisfy inequalities $u_i\le\sgg(i,A)$ for all~$i$
like the ones in the statement of the lemma.

This decomposition need not be sorted. However, because the values of $\sgg(i,A)$ are monotonically non-decreasing (Observation~\ref{obs:monsgg}),
swapping any two values of $u_i$ and $u_j$ that are out of order preserves the inequalities between these values and $\sgg(i,A)$ and $\sgg(j,A)$.
Since the sorted sequence of values $v_i$ can be obtained from the sequence $u_i$ by such swaps, it also obeys all the same inequalities.
\end{proof}

Using this characterization we can strengthen Lemma~\ref{lem:stf} to more explicitly describe a game strategy for reaching any given position.

\begin{corollary}[how to play to reach any single position]
Let $P$ be any reachable position in game $\agg(n,A)$.
Then the following strategy for playing the game reaches $P$:
\begin{itemize}
\item If $P$ contains more than one tile, first play the strategy recursively to reach a position with one nonempty cell, containing the largest tile value in $P$. Then continue recursively in the game $\agg(n-1,A)$ on the remaining cells to construct the remaining tiles of $P$.
\item If $P$ contains only one tile, of value $x$, let $y$ be the largest value in $A$ that is less than $v$. Play the strategy recursively to reach a position with two  nonempty cells, with values $y$ and $x-y$, and then in the final step of the recursive strategy merge these two values.
\end{itemize}
\end{corollary}

The correctness of this strategy follows easily by using Theorem~\ref{thm:xreach} to prove that each recursive goal within this strategy is itself reachable.

Putting the results of this section together, we have the following simple recurrence for computing $\sgg(n,A)$ and $\tgg(n,A)$:

\begin{theorem}[recurrence for single-tile and total-value reachability]
\label{thm:2048-recur}
Beginning with
\[
\sgg(0,A)=\tgg(0,A)=0,
\]
we may compute $\sgg(n,A)$ as the smallest
value in $A$ whose difference from the next larger value in $A$ is larger than $\tgg(n-1,A)$,
or $\infty$ if no such value exists.
We may compute
\[
\tgg(n,A)=\sgg(n,A)+\tgg(n-1,A)=\sum_{i=1}^n\sgg(i,A).
\]
\end{theorem}

\begin{proof}
The computation of $\sgg(n,A)$ follows from Lemma~\ref{lem:step-by-step}. By that lemma, each tile value up to the given value can be reached from its predecessor in~$A$, and each larger value cannot be reached.

The computation of $\tgg(n,A)$ follows from Theorem~\ref{thm:xreach}. By that lemma,
the tile values of any reachable position are individually dominated by the values in the reachable position that has one tile of each value $\sgg(i,A)$ for $i$ ranging from $1$ to $n$.
The sum in the formula gives the value of this position, which clearly obeys the stated recurrence.
\end{proof}

\begin{corollary}[termination if and only if gaps are unbounded]
For every $A$, the values of $\sgg(n,A)$ and $\tgg(n,A)$ are finite for all $n$ (and the game $\agg(n,A)$ necessarily terminates for all $n$) if and only if the gaps between consecutive members of $A$ are not bounded in size.
\end{corollary}

\begin{proof}
As a sum of values of $\sgg$, $\tgg$ is finite if and only if $\sgg$ is.
Additionally, $\tgg$ is strictly monotonically increasing, because it is always possible to add a single tile of value one to a reachable position in $\agg(n-1,A)$ and produce a higher-value reachable position in $\agg(n,A)$. Therefore, for larger and larger values of $n$, the formula for $\sgg(n,A)$ will require us to find correspondingly larger gaps in the sequence of values in $A$. This will be possible, leading to finite values of $\sgg(n,A)$ for all $n$, if and only if $A$ has gaps of unbounded size.
\end{proof}

\section{Making change}

The change-making problem involves making change for a given amount of money, using as few coins as possible from a given set of coin denominations. Most countries have coinage that allows the problem to be solved optimally by a greedy algorithm: to make change for a given amount of money $x$, first select the largest-valued coin whose value $y$ is less than or equal to $x$,
and then (if $x\ne y$) recursively solve the remaining change-making subproblem for the value $x-y$. However, greedy change-making is not always optimal. For instance, consider the situation of a cashier who is trying to make change in US money, for which the most commonly-used coin denominations are 1 cent (the penny), 5 cents (the nickel), 10 cents (the dime), and 25 cents (the quarter). To make change for 30 cents, the optimal choice would be the greedy choice, one quarter and one nickel. But if the change tray is out of nickels, so that the only coin values available are 1, 10, and 25 cents, the optimal choice would be three dimes, while the greedy algorithm would instead choose a quarter and five pennies, twice as many coins.

Optimal change-making is weakly NP-hard but has a pseudopolynomial time dynamic program that is often used as an example or an exercise in undergraduate algorithms classes~\cite{CorLeiRiv-09,GooTam-15}.
However, although there have also been studies on sets of coins that would lead to small solutions~\cite{Sha-MI-03} or on counting distinct ways of making change~\cite{BeySwi-CACM-73}, much of the research on change-making has focused on a different problem: for which coinage systems is the greedy algorithm optimal~\cite{AdaAda-EJC-10,CowCowSte-EJC-08,MagNemTro-OR-75,KozZak-TCS-94,Pea-ORL-05}? This can be tested in polynomial time~\cite{Pea-ORL-05}.

A particularly simple test (the Magazine--Nemhauser--Trotter one-shot test) determines whether a system of coins has a stronger property, that if the coins are sorted by value from smallest to largest, every prefix of this sorted sequence forms a set of coins for which greedy change-making is optimal. For each prefix let $x$ and $y$ be the largest and second-largest coins in the prefix;
then the one-shot test rounds $x$ up to an integer multiple $ky$ of $y$ and applies the greedy change-making algorithm to this number $ky$. If it uses more than $k$ coins, the greedy algorithm is suboptimal, but if every prefix uses this number of coins or fewer, then the greedy algorithm can be proven to be optimal for all prefixes~\cite{MagNemTro-OR-75}.
Following Cowen et al.~\cite{CowCowSte-EJC-08}, we call a system of coins that passes this test \emph{totally greedy}. Although the change-making problem is usually considered only for finite sets of coin denominations, the one-shot test and the prefix-greedy definition make sense equally well for infinite sets. For instance, the powers of two are totally greedy (the $i$th instance of the one-shot test uses one coin to represent the test value $2\cdot 2^{i-1}=2^i$) as are the Fibonacci numbers (the $i$th instance of the one-shot test uses two coins to represent the test value $2F_{i-1}=F_i+F_{i-3}$).

In connection with our analysis of abstract generalized 2048, we are interested in the behavior of the greedy algorithm on arbitrary coinage systems, regardless of whether the greedy algorithm is optimal for the system. The following quantity is of particular interest:

\begin{definition}[hard-to-change inputs to the greedy algorithm]
For any integer $n\ge 0$ and set of positive integer coin values $A$, we define $\ggg(n,A)$ to be the smallest integer $x$ that causes the greedy change-making algorithm to use at least $n$ coins.
\end{definition}

The following result is folklore; it is possible that it was first observed by Pillai in his 1930 study of greedy change-making for prime-number coin values~\cite{Pil-AUJ-30} but we have been unable to obtain a copy of his paper to check.

\begin{lemma}[recurrence for hard-to-change inputs]
\label{lem:ggg-recur}
We may compute $\ggg(n,A)$ using the recurrence
\[
\ggg(n,A)=\ggg(n-1,A)+x,
\]
where $x$ is the smallest member of $A$ such that the difference between $x$ and the next-larger member of $A$ exceeds $\ggg(n-1,A)$, and with the base case $\ggg(0,A)=0$.
\end{lemma}

\begin{proof}
The greedy algorithm will use $n$ or more coins on a given number $s$ if and only if $s$ has the form $t+u$ where $t$ is a member of $A$, $t+u$ is less than the next larger member of $A$ (so that the greedy algorithm begins by choosing $t$), and the greedy algorithm uses $n-1$ or more coins on $u$ (its recursive subproblem). The number $\ggg(n-1,A)+x$ has this form, with $t=x$ and $u=\ggg(n-1,A)$. It is the smallest number with this form, because any smaller value of $u$ would not cause the greedy algorithm to use $n-1$ or more coins on $u$, and any smaller value of $t$ with the same or larger value of $u$ would cause there to exist another member $r$ of $A$ in the range $t<r\le t+u$, preventing the greedy algorithm from starting by choosing~$t$.
\end{proof}

We are now ready to prove our min-max theorem relating 2048 to change-making:

\begin{theorem}[equality of 2048 and greedy change-making]
The maximum total value achieved in an $n$-cell abstract greedy 2048 game, $\tgg(n,A)$,
equals the minimum value that would cause the greedy change-making algorithm
to use $n$ or more coins, $\ggg(n,A)$.
\end{theorem}

\begin{proof}
By Theorem~\ref{thm:2048-recur} and Lemma~\ref{lem:ggg-recur}, both of these numbers are computed by the same recurrence with the same base case.
\end{proof}

\section{Specific sets of tile values}

The Python code in \autoref{tbl:total} takes as input a generator for a sorted sequence $A$ of tile values in an abstract generalized 2048 game (or of coin values in a greedy change-making problem), and returns a generator for the sequence of total tile values $\tgg(n,A)$ achievable with $n=1,2,3,\dots$ cells. It does so by computing, for each tile value in $A$, the gap between that value and the previous value, and when that gap is large enough using it to take a step in the recurrence for $\tgg(n,A)$. As can be seen from the code, the total space necessary (beyond that for generating $A$) consists only of a constant number of integer variables.
It is not possible to analyze the performance of this algorithm in terms of the variable $n$ without knowing more about the behavior of gaps in the sequence~$A$, but we can at least state that the time to generate all values $\tgg(n,A)$ that are below some threshold value $N$ is at most proportional to the time to generate all values in $A$ below the same threshold.

\begin{table}[ht]
\begin{lstlisting}
def Total(A):
    single,total = 1,0
    for tile in A:
        while tile > single + total:
            total += single
            yield total
        single = tile
\end{lstlisting}
\caption{Python code to generate the sequence of values $\tgg(n,A)$ from a generator for sequence $A$, using only a constant number of additional integer variables.}
\label{tbl:total}
\end{table}

For sequences that are totally greedy, the same algorithm will determine more strongly the smallest value that requires $n$ terms to represent as a sum of sequence values (not just as a greedy sum).
We ran this code using several different integer sequences~$A$, to determine for each one its corresponding sequence of maximum achievable total game values $\tgg(n,A)$. We identify each sequence using its code in the Online Encyclopedia of Integer Sequences (\href{https://oeis.org}{oeis.org}), a string of the form A$xxx$ where the $x$'s are decimal digits. Although some of these sequences would be problematic for games that combine tile values only in pairs (because their tile values cannot be reached by such pairwise combinations), this is not an issue for our abstract generalized 2048 game, which allows combinations of more than two tiles at once.

\begin{description}
\item[\oeis{A000040}]~\\
This is the sequence of prime numbers, $2, 3, 5, 7, 11,\dots$, in which we included also 1 (even though it is not prime) to make a valid set of tile or coin values.  It is not totally greedy.
When $A$ is this sequence, $\tgg(n,A)$ is Pillai's sequence~\cite{Pil-AUJ-30,LucTha-JTNB-09} \oeis{A066352} of the numbers $1,4,27,1354,401429925999155061,\dots.$
Because the gaps in the prime numbers grow so slowly, it has been estimated in the OEIS that the next number of this sequence would require hundreds of millions of digits.

\smallskip
\item[\oeis{A000045}]~\\
This is the sequence of Fibonacci numbers, $1, 2, 3, 5, 8, 13, \dots$, used in the 987 game. It is totally greedy.
When $A$ is this sequence, $\tgg(n,A)$ is the sequence \oeis{A027941} of numbers $F_{2n+1}-1=1, 4, 12, 33, 88, \dots$ of every other Fibonacci number, minus one.

\smallskip
\item[\oeis{A000079}]~\\
This is the sequence of powers of two, $2^i=1,2,4,8,\dots$, used in the 2048 game. It is totally greedy.
When $A$ is this sequence, $\tgg(n,A)$ is the sequence of Mersenne numbers \oeis{A000225}, $2^{n-1}-1=1,3,7,15,\dots$.

\smallskip
\item[\oeis{A000225}]~\\
This is the sequence of Mersenne numbers, $M_i=2^{i+1}-1=1,3,7,15,\dots$. It is also totally greedy,
because each prefix of the sequence passes the one-shot test according to the identity $3M_i=M_{i+1}+2M_{i-1}$. When $A$ is this sequence, $\tgg(n,A)$ is the sequence \oeis{A000325} of numbers $2^n-n=1,2,5,12,27,\dots$ which is not totally greedy ($3\cdot 12=27+5+2+2$ is expanded to four coins, not three, by the greedy algorithm, failing the one-shot test).

\smallskip
\item[\oeis{A005153}]~\\
This is the sequence of practical numbers $1, 2, 4, 6, 8, 12, 16,\dots$ discussed in the introduction.
It can be generated by using a variation of the sieve of Eratosthenes to generate the factorizations of each positive integer, and then using an efficient test of Stewart and Sierpinski~\cite{Ste-AJM-54,Sie-AMPA-55} to determine from each factorization whether each integer is practical.
When $A$ is this sequence, $\tgg(n,A)$ is a sequence beginning $1,3,11,191$, not in the OEIS.
Because the gaps in the sequence of practical numbers are (like the gaps in the primes) slowly growing, the next number in the sequence should be quite large.

\smallskip
\item[\oeis{A003586}]~\\
This is the sequence of $3$-smooth numbers $1, 2, 3, 4, 6, 8, 9, \dots$ (the numbers having only $2$ or $3$ as prime factors), discussed in the introduction. It is not totally greedy. When $A$ is this sequence, $\tgg(n,A)$ is the sequence \oeis{A296840}: $1, 5, 23, 185, 1721, 15545, 277689,\dots$.

\smallskip
\item[\oeis{A029744}]~\\
This is the sequence of numbers $2^i$ or $3\cdot 2^i=1, 2, 3, 4, 6, 8, 12,\dots$ used in the game Threes. It is totally greedy. When $A$ is this sequence, $\tgg(n,A)$ is the sequence \oeis{A002450} of numbers $(4^{n+1}-1)/3=1,5,21,85,341,\dots$.

\smallskip
\item[\oeis{A126684}]~\\
This is the sequence of numbers $1, 2, 4, 5, 8, 10, 16, 17, 20, 21, 32, \dots$
whose binary representations have either all even bit positions zero or all odd bit positions zero.
It gives perhaps the most extreme example of the distinction between optimal and greedy change-making:
for a system of coins with these values, any amount of change can be made with at most two coins,
and the $\Theta(n^2)$ growth rate of this sequence is the fastest possible for this two-coin property. However, greedy change-making will typically use more than two coins. For instance, although $13$ can be represented as the sum of two sequence members $8+5$, its greedy representation is $10+2+1$. When $A$ is this sequence, $\tgg(n,A)$ is the sequence \oeis{A302757} of numbers $1, 3, 13, 55, 225, 907, 3637, \dots$, which grows exponentially according to the recurrence $a_n=4a_{n-1}+2n-5$.
\end{description}

\noindent
The tile values $1,2,3,5,10,20,40,80,\dots$ in the game Fives are not listed in the OEIS, but the maximum achievable values $\tgg(n,A)$ form the sequence \oeis{A052549} of numbers  $1, 4, 9, 19, 39, 79$. They have the formula $\lfloor 5\cdot 2^{n-2}-1\rfloor$.

A similar analysis could be applied to many other sequences, yielding new sequences not already part of the OEIS. For instance, sequences  \oeis{A296840} and \oeis{A302757}  were added to the OEIS as a result of our investigations, not having been studied before.

\section{Discussion}

We have described an abstract version of the game 2048 that eliminates the geometry and other complicating factors of the game, allowing us to provide a complete analysis of our abstract game for any set of allowable tile values and any number of cells. We proved a min-max theorem equating the maximum total tile value that can be achieved in this game with the minimum value that would cause a greedy change-making algorithm, using coins of the same value as the tiles, to use the same number of coins as the number of cells in the game. Finally, we showed how to compute the values from this theorem by a streaming algorithm that uses only a constant number of integer variables beyond the requirements of generating the tile value sequence itself, and used our implementation to compute the sequences of maximum game values for several choices of allowable tile value sets.

It would be of interest to understand  in more detail for which non-abstract 2048-like games this analysis is tight or nearly tight, and for which it fails to capture the game dynamics and produces a bound on the total game value that is large compared to the actual achievable value. For instance, experience with 2048 and 987 suggests that, in those games, a strategy close to that of the abstract game can usually be followed, leading to total game values similar to what could be achieved in the abstract game. On the other hand, in Threes, the inability to add some pairs of game tiles such as $1+3$, even when the sum would be another allowable tile value, may cause this game's maximum achievable value to be closer to $2^n$ than to the $(4^{n+1}-1)/3$ formula for the total score achievable on the corresponding abstract game. We leave such questions open for future research.

Additionally, some variants of 2048 are not amenable to our analysis. These include 2048 Circle of Fifths, a game based on the circle of fifths in music theory whose tile values involve modular arithmetic~\cite{Hugo}, and 2048 Numberwang, in which the tile combinations that are allowed on each move vary randomly~\cite{Huang}. Developing a theoretical analysis of these games could be fun.

\bibliographystyle{plainurl}
\raggedright
\bibliography{2048}
\end{document}